\documentclass[11pt]{amsart}

\usepackage{amsmath,amssymb, mathrsfs, dsfont, amsthm, color}
\usepackage{bbm}
\usepackage{stmaryrd}
\usepackage{graphicx}
\usepackage{paralist}
\setdefaultenum{(i)}{}{}{}
\usepackage[margin=1.25in]{geometry}
\usepackage{tikz}

\newtheorem{theorem}{Theorem}[section]
\newtheorem{assumption}[theorem]{Assumption}

\newtheorem{corollary}[theorem]{Corollary}

\newtheorem{definition}[theorem]{Definition}

\newtheorem{lemma}[theorem]{Lemma}
\newtheorem{proposition}[theorem]{Proposition}

\theoremstyle{definition}

\newtheorem{remark}[theorem]{Remark}

\numberwithin{equation}{section}

\renewcommand\hat[1]{\widehat{#1}}

\newcommand\RR{\mathbb R}

\newcommand\Z{\mathcal Z}

\newcommand\one{\mathds1}
\newcommand\interior{\operatorname{int}}

\newcommand\set[1]{\left\{#1\right\}}

\newcommand\sets[2]{\set{#1\,:\,#2}}

\newcommand\cadlag{c\`adl\`ag{} }

\newcommand\esssup{\operatorname{ess\;sup}}

\title{On the existence of shadow prices}
\thanks{The authors are grateful to Bruno Bouchard, Christoph Czichowsky, Paolo Guasoni, Ioannis Karatzas, Marcel Nutz, Mark P.\ Owen, and Walter Schachermayer for fruitful discussions, and also acknowledge the constructive comments of two anonymous referees. The second author thanks the ``Chaire Les Particuliers Face aux Risques'', Fondation du Risque (Groupama-ENSAE-Dauphine), the GIP-ANR ``Risk'' project for their support. The fourth author was partially supported by the National Centre of Competence in Research ÒFinancial Valuation and Risk ManagementÓ (NCCR FINRISK), Project D1 (Mathematical Methods in Financial Risk Management), of the Swiss National Science Foundation (SNF)}


\author[G.\ Benedetti]{Giuseppe Benedetti}
\address{CREST and Universit\'e Paris-Dauphine \endgraf Place du Mar\'echal de Lattre de Tassigny, FR-75775 Paris, France}
\email{giuseppe.benedetti@ensae.fr}
\author[L.\ Campi]{Luciano Campi}
\address{Universit\'e Paris 13, Laboratoire Analyse, G\'eom\'etrie et Applications, and CREST\endgraf 99, Avenue Jean-Baptiste Cl\'ement, FR-93430 Villetaneuse, France}
\email{campi@math.univ-paris13.fr}
\author[J.\ Kallsen]{Jan Kallsen}
\address{Christian-Albrechts-Universit\"at zu Kiel, Mathematisches Seminar\endgraf Westring 383, D-24118 Kiel, Germany}
\email{kallsen@math.uni-kiel.de}
\author[J.\ Muhle-Karbe]{Johannes Muhle-Karbe}
\address{ETH Z\"urich, Departement f\"ur Mathematik\endgraf R\"amistrasse 101, CH-8092 Z\"urich, Switzerland}
\email{johannes.muhle-karbe@math.ethz.ch}
\date{June 11, 2012}

\begin{document}

\maketitle

\begin{abstract}
For utility maximization problems under proportional transaction costs, it has been observed that the original market with
transaction costs can sometimes be replaced by a frictionless {\em shadow market} that yields the same optimal strategy and utility.
However, the question of whether or not this indeed holds in generality has remained elusive so far.
In this paper we present a counterexample which shows that shadow prices may
fail to exist. On the other hand, we prove that short selling constraints are a sufficient condition  to warrant their existence,
even in very general multi-currency market models with possibly discontinuous bid-ask-spreads.
\end{abstract}

\section{Introduction}

Transaction costs have a severe impact on portfolio choice: If securities have to be bought for an ask price which is higher than the bid price one receives for selling them, then investors are forced to trade off the gains and costs of rebalancing.\footnote{For example, Liu and Loewenstein \cite{liu.loewenstein.02} reckon that ``even small transaction costs lead to dramatic changes in the optimal behavior for an investor: from continuous trading to virtually buy-and-hold strategies.''} Consequently, utility maximization under transaction costs has been intensely studied in the literature. We refer the reader to Campi and Owen~\cite{campi.owen.11} for general existence and duality results, as well as a survey of the related literature.

It has been observed that the original market with transaction costs can sometimes be replaced by a fictitious frictionless ``shadow market'', that yields the same optimal strategy and utility. If such a shadow price exists, then transaction costs do not lead to qualitatively new effects for portfolio choice, as their impact can be replicated by passing to a suitably modified frictionless market. Starting from \cite{kallsen.muhlekarbe.10}, shadow prices have recently also proved to be useful for solving optimization problems in concrete models, see \cite{kuehn.stroh.10,gerhold.al.10,herzegh.prokaj.11,gerhold.al.11}. However, unlike in the contexts of local risk minimization \cite{lamberton.al.98}, no-arbitrage \cite{guasoni.al.08,jouini.kallal.95,schachermayer.04}, and superhedging \cite{cvitanic.al.99} (also cf.\ \cite{kabanov.safarian.09} for an overview) --- the question of whether or not such a least favorable frictionless market extension indeed exists has only been resolved under rather restrictive assumptions so far.

More specifically, Cvitani\'c and Karatzas~\cite{cvitanic.karatzas.96} answer it in the affirmative in an It\^o process setting, however, only under the assumption that the minimizer in a suitable dual problem exists and is a martingale.\footnote{That is, a \emph{consistent price system} in the terminology of Schachermayer~\cite{schachermayer.04}.} Yet, subsequent work by Cvitani\'c and Wang \cite{cvitanic.wang.01} only guarantees existence of the minimizer in a class of supermartingales. Hence, this result is hard to apply unless one can solve the dual problem explicitly.

A different approach leading closer to an existence result is provided by Loewenstein~\cite{loewenstein.00}. Here, the existence of shadow prices is established for continuous bid-ask prices in the presence of short selling constraints. In contrast to Cvitani\'c and Karatzas~\cite{cvitanic.karatzas.96}, Loewenstein~\cite{loewenstein.00} constructs his shadow market directly from the primal rather than the dual optimizer. However, his analysis is also based on the assumption that the starting point for his analysis, in this case his constrained primal optimizer, actually does exist.

Finally, Kallsen and Muhle-Karbe~\cite{kallsen.muhlekarbe.11} show that shadow prices always exist for utility maximization problems in finite probability spaces. But, as usual in Mathematical Finance, it is a delicate question to what degree this transfers to more general settings.

The present study contributes to this line of research in two ways. On the one hand, we present a counterexample showing that shadow prices do not exist in general without further assumptions. On the other hand, we establish that Loewenstein's approach can be used to come up with a positive result, even in Kabanov's~\cite{kabanov.99} general multi-currency market models with possibly discontinuous bid-ask-spreads:

\begin{theorem}\label{mainresult}
In the general multicurrency setting of Section 3, a shadow price always exists, \emph{if} none of the assets can be sold short.
\end{theorem}

The crucial assumption -- which is violated in our counterexample -- is the prohibition of short sales for all assets under consideration. Like Loewenstein~\cite{loewenstein.00}, we construct our shadow price from the primal optimizer. Existence of the latter is established by extending the argument of Guasoni~\cite{guasoni.02} to the general setting considered here. This is done by making use of a compactness result for predictable finite variation processes established in Campi and Schachermayer~\cite{campi.schachermayer.06}.

The paper is organized as follows. In Section \ref{counterexample}, we present our counterexample, which is formulated in  simple setting with one safe and one risky asset for better readability. Afterwards, the general multi-currency framework with transaction costs and short selling constraints is introduced. In this setting, we then show that shadow prices always exist. Finally, Section \ref{conclusion} concludes.

\section{A counterexample}  \label{counterexample}
In this section, we show that even a simple discrete-time market can fail to exhibit a shadow price,\footnote{Cf.\ Definition \ref{def:shadow} below for the formal definition} if the unconstrained optimal strategy involves short selling. Consider a market with one safe and one risky asset traded at the discrete times $t=0,1,2$: The bid and ask prices of the safe asset are equal to $1$ and the bid-ask spread\footnote{In the general multicurrency notation introduced below, this corresponds to $[1/\pi^{21},\pi^{12}]$, where $\pi^{ij}$ denotes the number of units of asset $i$ for which the agent can buy one unit of asset $j$.} $[\underline S,\overline S]$ of the risky asset is defined as follows. The bid prices are deterministically given by
$$\underline S_0=3,\quad  \underline S_1=2,\quad \underline S_2=1,$$
while the ask prices satisfy $\overline S_0=3$ and
\begin{eqnarray*}
 \mathbb P(\overline S_1=2)&=&1-2^{-n},\\
\mathbb P(\overline S_1=2+k)&=&2^{-n-k}, \quad k=1,2,3,\dots,
\end{eqnarray*}
where $n\in\mathbb N$ is chosen big enough for the subsequent argument to work.
Moreover, we set
\begin{eqnarray*}
 \mathbb P(\overline S_2=3+k|\overline S_1=2+k)&=&2^{-n-k},\\
\mathbb P(\overline S_2=1|\overline S_1=2+k)&=&1-2^{-n-k}
\end{eqnarray*}
for $k=0,1,2,\dots$ The corresponding bid-ask process is illustrated in Figure \ref{fig:1}. One readily verifies that a strictly consistent price system exists in this market. Now, consider the maximization of expected logarithmic utility from terminal wealth, where the maximization takes place over all self-financing portfolios that liquidate at $t=2$ after starting from an initial position of $-1$ risky and $4$ safe assets at time $0$.

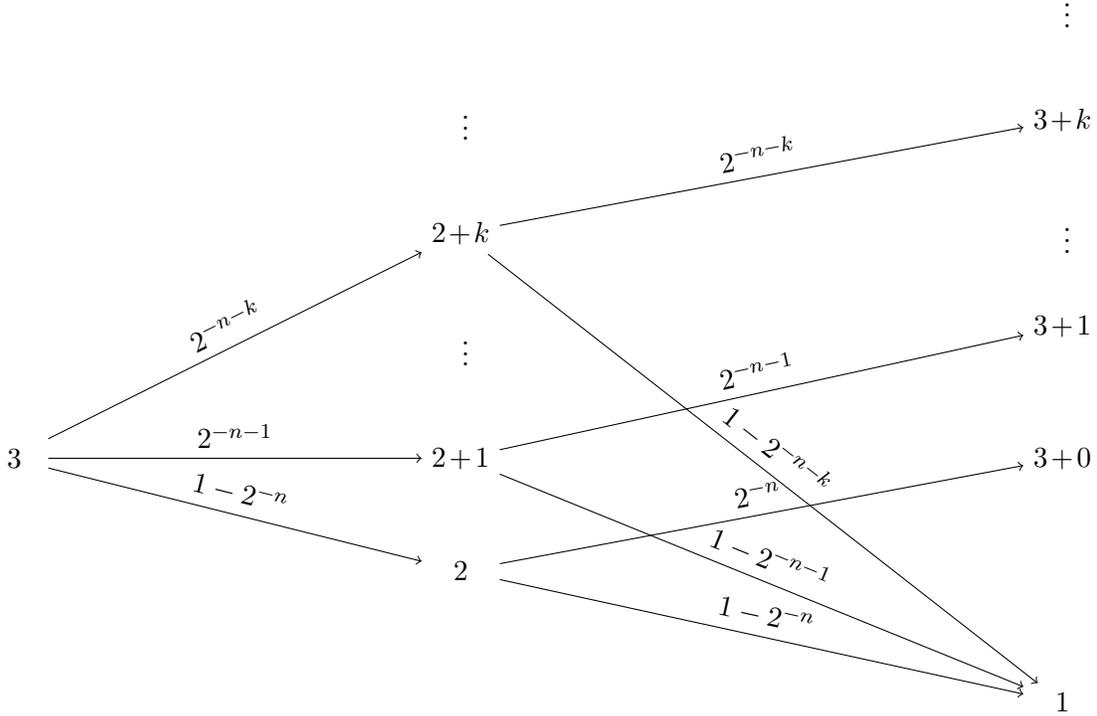
\begin{figure}[htbp]
 \tikzstyle{bag} = [text width=2em, text centered]
\tikzstyle{end} = []
\begin{center}
\begin{tikzpicture}[sloped]
 \node (a) at ( 0,0) [bag] {$\ 3$};
 \node (b1) at ( 6,-1.5) [bag] {$2$};
 \node (b2) at ( 6,0) [bag] {$2+1$};
 \node (b3) at ( 6,1.5) [bag] {\ \vdots};
 \node (b4) at ( 6,3) [bag] {$2+k$};
 \node (b5) at ( 6,4.5) [bag] {\ \vdots};
\node (c1) at ( 14,-3.25) [bag] {1};
\node (c2) at ( 14,0) [bag] {$3+0$};
\node (c3) at ( 14,1.75) [bag] {$3+1$};
\node (c4) at ( 14,3) [bag] {\ \vdots};
\node (c5) at ( 14,4.5) [bag] {$3+k$};
\node (c6) at ( 14,6) [bag] {\ \vdots};

\draw [->] (a) to node [above] {$1-2^{-n}$} (b1);
\draw [->] (a) to node [above] {$2^{-n-1}$} (b2);
\draw [->] (a) to node [above] {$2^{-n-k}$} (b4);
\draw [->] (b1) to node [above] {$1-2^{-n}$} (c1);
\draw [->] (b1) to node [above] {$2^{-n}$} (c2);
\draw [->] (b2) to node [above] {$1-2^{-n-1}$} (c1);
\draw [->] (b2) to node [above] {$2^{-n-1}$} (c3);
\draw [->] (b4) to node [above] {$1-2^{-n-k}$} (c1);
\draw [->] (b4) to node [above] {$2^{-n-k}$} (c5);
\end{tikzpicture}
\caption{Illustration of the ask price in the counterexample. The corresponding bid price decreases deterministically from 3 to 2 to 1.}\label{fig:1}
\end{center}
\end{figure}

It is not hard to determine the optimal trading strategy in this setup.
Buying a positive amount of stock at time 0 is suboptimal because
expected gains would be negative. Consequently, zero holdings are preferable to positive ones.
A negative position in the first period, on the other hand, is impossible as well
because it  may lead to negative terminal wealth.
Hence it is optimal to do nothing at time $0$,
i.e., the optimal numbers $\hat V=(\hat{V}^1,\hat{V}^2)$ of safe and risky assets satisfies $\hat V^2_0=0$. In the second period,  a positive stock holding would be again suboptimal
because prices are still falling on average. By contrast, building up a negative position
is worthwhile. The stock can be sold short at time $1$ for $\underline S_1=2$ and
it can be bought back at time $2$ for $\overline S_2=1$ with overwhelming probability and
for $\overline S_2=3$ resp.\ $3+k$  with very small probability.
Consequently, the optimal strategy satisfies $\hat V^2_1<0$ in any state.

If a shadow price process $(1,\widetilde S)$ for this market exists, then $\tilde S$ must coincide with the bid resp.\ ask price if a transaction takes place in the optimal strategy $\hat V$. Otherwise, one could achieve strictly higher utility trading $\tilde{S}$, by performing the same purchases and sales at sometimes strictly more favorable prices. Consequently, we must have
$\widetilde S_0=\overline S_0=3$, $\widetilde S_1=\underline S_1=2$,
$\widetilde S_2=\overline S_2$.
However, $(1,\widetilde S)$ cannot be a shadow price process. Indeed, $\widetilde S$ is decreasing deterministically by $1$ in the first period and would allow for unbounded expected utility and in fact even for arbitrage.

\begin{remark}\label{rem:counterconstraints}
It is important to note that the market discussed in this section \emph{does} admit a shadow price if one imposes short selling constraints. In this case, it is evidently optimal not to trade at all in the original market with transaction costs: Positive positions are not worthwhile because prices are always falling on average, whereas negative positions are ruled out by the constraints. In this market \emph{any} supermartingale $(1,\widetilde S)$ with values in the bid-ask spread, i.e., $\widetilde S_0=3$, $2 \leq \widetilde S_1\leq \overline S_1$ and $\widetilde S_2=1$, is a shadow price (showing that even if a shadow price exists, it need not be unique). Indeed, Jensen's inequality yields that positive positions are suboptimal, and negative holdings are again prohibited by the constraints. Hence it is optimal not to trade at all, as in the original market with transaction costs. This confirms that $(1,\widetilde S)$ is indeed a shadow price if short selling is ruled out. However, it is clearly not a shadow price in the unconstrained market, as it would allow for obvious arbitrage.
\end{remark}

\section{The General Multi-Currency Model}\label{sec:general}

Henceforth, we work in the general transaction cost framework of Campi and Schachermayer \cite{campi.schachermayer.06}, with slight modifications to incorporate short selling constraints. We describe here the main features of the model, but refer to the original paper for further details. For any vectors $x,y$ in $\mathbb R^d$, we write $x \succeq y$ if  $x-y \in \mathbb R^d_+$ and $xy$ for the Euclidean scalar product.

Let $(\Omega,(\mathcal F_t)_{t\in[0,T]},\mathbb P)$ be a filtered probability space satisfying the usual conditions and supporting all processes appearing in this paper; the initial $\sigma$-field is assumed to be trivial.

We consider an agent who can trade in $d$ assets according to some \emph{bid-ask matrix} $\Pi=(\pi^{ij})_{1\le i,j\le d}$, where $\pi^{ij}$ denotes the number of units of asset $i$ for which the agent can buy one unit of asset $j$. To recapture the notion of currency exchanges, one naturally imposes as in \cite{schachermayer.04} that:
\begin{enumerate}
\item[(i)] $\pi^{ij}\in (0,\infty)$ for every $1\le i,j\le d$;
\item[(ii)] $\pi^{ii}=1$ for every $1\le i\le d$;
\item[(iii)] $\pi^{ij}\le\pi^{ik}\pi^{kj}$ for every $1\le i,j,k\le d$.
\end{enumerate}
The first condition means that the bid-ask prices of each asset in terms of the others are positive and finite, while the interpretation of the second is evident. The third implies that direct exchanges are not dominated by several successive trades. In the spirit of \cite{kabanov.99}, the entries of the bid-ask matrix can also be interpreted in terms of the prices $S^1,\ldots,S^d$ of the assets and proportional transaction costs $\lambda^{ij}$  for exchanging asset $i$ into asset $j$, by setting
$$\pi^{ij}=(1+\lambda^{ij})\frac{S^j}{S^i}.$$
We will use both notations in the sequel, one being shorter and the other providing a better financial intuition. Given a bid-ask matrix $\Pi$, the \emph{solvency cone} $K(\Pi)$ is defined as the convex polyhedral cone in $\RR^d$ spanned by the canonical basis vectors $e_i$, $1\le i\le d$, of $\RR^d$, and the vectors $\pi^{ij}e_i-e_j$, $1\le i,j\le d$.\footnote{$K(\Pi)$ contains precisely the \emph{solvent} portfolios that can be liquidated to zero by trading according to the bid-ask matrix $\Pi$ and possibly throwing away positive asset holdings.} The convex cone $-K(\Pi)$ should be interpreted as those portfolios available at price zero in a market without short selling constraints. Given a cone $K$, its (positive) polar cone is defined by
\[ K^* = \sets{w\in\RR^d}{ vw\ge0,\forall v\in K}. \]

We now introduce randomness and time in the model. An adapted, \cadlag process $(\Pi_t)_{t\in[0,T]}$ taking values in the set of bid-ask matrices will be called a \emph{bid-ask process}. Once a bid-ask process $(\Pi_t)_{t\in[0,T]}$ has been fixed, one can drop it from the notation by writing $K_\tau$ instead of $K(\Pi_\tau)$ for any stopping time $\tau$, coherently with the framework introduced above.
In accordance with the framework developed in \cite{campi.schachermayer.06} we make the following technical assumption throughout the paper. It means basically that no price changes take place at time $T$, which serves only as a date for liquidating the portfolio. This assumption can be relaxed via a slight modification of the model (see \cite[Remark 4.2]{campi.schachermayer.06}). For this reason, we shall not explicitly mention it in the following.
\begin{assumption}\label{ass:contPi}
 $\mathcal F_{T-}=\mathcal F_T$ and $\Pi_{T-}=\Pi_T$ a.s.
\end{assumption}

In markets with transaction costs, consistent price systems play a role similar to martingale measures for the frictionless case (compare, e.g., \cite{schachermayer.04,guasoni.al.08,kabanov.safarian.09}). For utility maximization, this notion has to be extended, just as it is necessary to pass from martingale measures to supermartingale densities in the frictionless case \cite{kramkov.schachermayer.99}:

\begin{definition}
An adapted, $\RR_+^d\setminus\{0\}$-valued, \cadlag supermartingale $Z=(Z_t)_{t\in[0,T]}$ is called a \emph{supermartingale consistent price system} (supermartingale-CPS) if $Z_t\in K_t^*$ a.s.\ for every $t\in[0,T]$. Moreover, $Z$ is called a \emph{supermartingale strictly consistent price system} (supermartingale-SCPS) if it satisfies the following additional condition: for every $[0,T]\cup\set{\infty}$-valued stopping time $\tau$, we have $Z_\tau\in\interior(K_\tau^*)$ a.s.\ on $\set{\tau<\infty}$, and for every predictable $[0,T]\cup\set{\infty}$-valued stopping time $\sigma$, we have $Z_{\sigma-}\in\interior(K_{\sigma-}^*)$ a.s.\ on $\set{\sigma<\infty}$.
The set of all supermartingale-(S)CPS is denoted by $\Z_{\sup}$ (resp.\ $\Z^s_{\sup}$).
\end{definition}
As in \cite{campi.schachermayer.06}, trading strategies are described by the numbers of physical units of each asset held at time $t$:

\begin{definition}  \label{SF}
An $\RR^d$-valued process $V = (V_t)_{t\in[0,T]}$ is called a \emph{self-financing portfolio process} for the process $K$ of solvency cones if it satisfies the following properties:
 \begin{enumerate}
   \item It is predictable and a.e.\ (not necessarily right-continuous) path has finite variation.
   \item For every pair of stopping times $0\le\sigma\le\tau\le T$, we have
   \[ V_\tau-V_\sigma \in -K_{\sigma,\tau}, \]
 \end{enumerate}
 where $K_{s,t}(\omega)$ denotes the closure of $\mathrm{cone}\{K_r(\omega), s\leq r<t\}$.

 A self-financing portfolio process $V$ is called \emph{admissible} if it satisfies the \emph{no short selling constraint} $V\succeq 0$.
 \end{definition}

We need some more notation related to such processes. For any predictable process of finite variation $V$, we can define its continuous part $V^c$ and its left (resp.\ right) jump process $\Delta V_t := V_t - V_{t-}$ (resp.\ $\Delta_+ V_t := V_{t+} -V_t$), so that $V=V^c + \Delta V + \Delta_+ V$. The continuous part, $V^c$, is itself of finite variation, so we can define its Radon-Nykodim derivative $\dot{V}^c _t$ with respect to its total variation process $\textrm{Var}_t (V^c)$, for all $t\in [0,T]$ (see \cite[Section 2]{campi.schachermayer.06} for details).

We will work under the following no-arbitrage assumption, which is the analogue of the existence of a supermartingale density in frictionless markets.
\begin{assumption}  \label{ass_CPS}
$\mathcal Z^s_{\sup} \neq \emptyset$.
\end{assumption}

We now turn to the utility maximization problem. Here, we restrict our attention without loss of generality to admissible and self-financing portfolio processes that start out with some initial endowment $x \in \mathbb{R}^d_+ \setminus \{0\}$, and such that $V_T$ is nonzero only in the first component (that is, the agent liquidates his wealth to the first asset at the final date). The set of those processes is denoted by $\mathcal A^{x,ss}$, and the set
\[ \mathcal A_T^{x,ss}:=\sets{V^1_T}{V\in\mathcal A^{x,ss}} \]
consists of all terminal payoffs (in the first asset) attainable at time $T$ from initial endowment~$x$. Moreover, the set
\[ \mathcal A_{T-}^{x,ss}:=\sets{V_{T-}}{V\in\mathcal A^{x,ss}} \]
contains the pre-liquidation values of admissible strategies.

The utility maximization problem considered in this paper is the following:
\begin{equation}	\label{trans_primal}
J(x):= \sup_{f \in \mathcal{A}^{x,ss}_T}\mathbb E[U(f)].
\end{equation}
Here, $U: (0,\infty) \to \mathbb R$ is a utility function in the usual sense, i.e., a strictly concave, increasing, differentiable function satisfying \begin{enumerate}
\item the \emph{Inada conditions} $\lim_{x\downarrow 0} U'(x) = \infty$ and $\lim_{x\uparrow \infty} U'(x) = 0$, and
\item the condition of \emph{reasonable asymptotic elasticity} (RAE): $\limsup_{x\to \infty} \frac{xU'(x)}{U(x)} <1$.
\end{enumerate}
We write $U^*(y)=\sup_{x>0}[U(x)-xy]$, $y>0$ for the conjugate function of $U$, and $I := (U')^{-1}$ for the inverse function of its derivative. To rule out degeneracies, we assume throughout that the maximal expected utility is finite:

\begin{assumption} \label{assfinite}
$J(x)= \sup_{f \in \mathcal{A}^{x,ss}_T}\mathbb E[U(f)]<\infty$.
\end{assumption}

A unique maximizer for the utility maximization problem \eqref{trans_primal} indeed exists:\footnote{In the absence of constraints, similar existence results have been established for increasingly general models of the bid-ask spread by \cite{cvitanic.karatzas.96, deelstra.al.01,bouchard.02,guasoni.02,campi.owen.11}.}

\begin{proposition}\label{prop:existence}
Fix an initial endowment $x\in \mathbb{R}^d_+ \setminus \{0\}$. Under Assumptions \ref{ass_CPS} and \ref{assfinite}, the utility maximization problem \eqref{trans_primal} admits a unique solution $\hat f \in \mathcal{A}^{x,ss}_T$.
\end{proposition}

\begin{proof}
\emph{Step 1}: The compactness result for predictable finite variation processes established in  \cite[Proposition 3.4]{campi.schachermayer.06} also holds in our setting where, in particular, the existence of a strictly consistent price system is replaced by the weaker Assumption \ref{ass_CPS}. To see this, it suffices to show that the estimate in \cite[Lemma 3.2]{campi.schachermayer.06} is satisfied under Assumption \ref{ass_CPS}; then, the proof of \cite[Proposition 3.4]{campi.schachermayer.06} can be carried through unchanged. To this end, we can use the arguments in \cite[Section 3]{campi.schachermayer.06} with the following minor changes:
\begin{enumerate}
\item First, notice that in the proof of \cite[Lemma 3.2]{campi.schachermayer.06} the martingale property of $Z$ is only used to infer that this strictly positive process satisfies $\inf_{t\in [0,T]} \|Z_t\|_d>0$ a.s., and this remains true for supermartingales by \cite[VI, Theorem 17]{dellacherie.82}.
\item \cite[Lemma 3.2]{campi.schachermayer.06} is formulated for strategies starting from a zero initial position, but can evidently be generalized to strategies starting from any initial value $x$. Moreover, in our case the admissible strategies are all bounded from below by zero due to the short selling constraints. Consequently, \cite[Lemma 3.2]{campi.schachermayer.06} holds under the weaker Assumption \ref{ass_CPS} for $V\in\mathcal A^{x,ss}$ and, in our case, \cite[Equation (3.5)]{campi.schachermayer.06} reads as $\mathbb{E}_Q[\textrm{Var}_T(V)]\leq C\|x\|$ for a constant $C \geq 0$ and an equivalent probability $Q$.
\end{enumerate}
\emph{Step 2}: Pick a maximizing sequence $(V^n)_{n\geq 1}\in \mathcal A^{x,ss}$ for \eqref{trans_primal} such that $\mathbb E[U(V_T ^n )] \to J(x)$ as $n\to \infty$. By Step 1, we can assume (up to a sequence of convex combinations) that $V^n _T \to V^0 _T$ a.s.\ for some $V^0 \in \mathcal A^{x,ss}$. The rest of the proof now proceeds as in \cite[Theorem 5.2]{guasoni.02}: by means of RAE assumption, we can prove that $\lim_{n\to \infty} \mathbb E[U(V_T ^n )]  \leq \mathbb E[U(V_T ^0 )]$, implying that $V^0$ is the optimal solution to \eqref{trans_primal}.  Uniqueness follows from the strict concavity of $U$.
\end{proof}

Consider now a supermartingale-CPS $Z\in\mathcal Z_{\sup}$. By definition, $Z$ lies in the polar $K^*$ of the solvency cone (we omit dependence on time for clarity); since $Z \neq 0$ this implies in particular that all components of $Z$ are strictly positive. Moreover, taking any asset, say the first one, as a numeraire, it means that
\begin{equation}	\label{spread}
\frac{1}{1+\lambda^{i1}}\frac{S^i}{S^1}\leq\frac{Z^i}{Z^1}\leq(1+\lambda^{1i})\frac{S^i}{S^1},
\end{equation}
for any $i=1,\ldots,d$. In other words, the frictionless price process
$$S^Z:=Z/Z^1$$
evolves within the corresponding bid-ask spread. This implies that the terms of trade in this frictionless economy are at least as favorable for the investor as in the original market with transaction costs. For  $S^Z$, we use the standard notion of a self-financing strategy:\footnote{In particular, we do \emph{not} restrict ourselves to finite variation strategies here.}

\begin{definition}  \label{SF_shadow}
Let $S^Z:=Z/Z^1$ for some $Z\in\mathcal Z_{\sup}$. Then, a predictable, $\RR^d$-valued, $S^Z$-integrable process $V = (V_t)_{t\in[0,T]}$ is called a \emph{self-financing portfolio process} in the frictionless market with price process $S^Z$, if it satisfies
\begin{equation}\label{eq:sf_shadow}
V_tS^Z_t= xS^Z_0+\int_0^t V_u dS^Z_u, \quad t\in[0,T].
\end{equation}
It is called \emph{admissible} if it additionally satisfies the no short selling constraint $V\succeq 0$. We sometimes write $Z$-admissible to stress the dependence on a specific $Z\in\mathcal Z_{\sup}$.

The set $\mathcal{A}^{x,ss}_T(Z)$ consists of all payoffs $V_T S^Z_T$ that are attained by a $Z$-admissible strategy $V$ starting from initial endowment~$x \in \mathbb{R}^d_+ \setminus \{0\}$.
\end{definition}

This notion is indeed compatible with Definition \ref{SF}, in the sense that any payoff in the original market with transaction costs can be dominated in the potentially more favorable frictionless markets evolving within the bid-ask spread:

\begin{lemma}\label{lem:dom}
Fix $Z\in\mathcal Z_{\sup}$. For any admissible strategy $V$ in the sense of Definition \ref{SF} there is a strategy $\tilde V\succeq V$ which is $Z$-admissible in the sense of Definition \ref{SF_shadow}.
\end{lemma}
\begin{proof}
Let $V$ be self-financing in the original market with transaction costs in the sense of Definition \ref{SF}. Then, since $V$ is of finite variation, applying the integration by parts formula as in the proof of \cite[Lemma 2.8]{campi.schachermayer.06} yields
\begin{equation}\label{eq:upperbound}
  \begin{split}
    V_t S^Z_t&=xS^Z_0+\int_0^t V_u dS^Z_u+\int_0^t S^Z_u \dot{V}^c_u d\textrm{Var}_u(V^c)+\displaystyle\sum_{u\leq t} S^Z_{u-}\Delta V_u+\displaystyle\sum_{u< t} S^Z_{u}\Delta_+ V_u\\
    &\leq xS^Z_0+\int_0^t V_u dS^Z_u
  \end{split}
\end{equation}
because, since $Z\in K^*$ and $Z^1>0$ imply $S^Z\in K^*$, we can use \cite[Lemma 2.8]{campi.schachermayer.06} to get
\[ \int_0^t S^Z_u \dot{V}^c_u d\textrm{Var}_u(V^c)+\displaystyle\sum_{u\leq t} S^Z_{u-}\Delta V_u+\displaystyle\sum_{u< t} S^Z_{u}\Delta_+ V_u \leq 0  .\]
Now, define the portfolio process $\tilde V$ by
$$\tilde V^1_t:=xS^Z_0+\int_0^t V_u dS^Z_u-\displaystyle\sum_{i=2}^d V^i_t S^{Z,i}_t=xS^Z_0+\int_0^{t-} V_u dS^Z_u-\displaystyle\sum_{i=2}^d V^i_t S^{Z,i}_{t-},$$
and $\tilde V^i_t:=V^i_t$ for $i=2,\ldots,d$. Then, $\tilde V$ is self-financing in the sense of Definition \ref{SF_shadow} by construction. Moreover, again by definition and due to \eqref{eq:upperbound}, we have $$\tilde V_t S^Z_t = xS^Z_0+\int_0^t V_u dS^Z_u \geq  V_t S^Z_t$$
and in turn $\tilde V\succeq V$.
\end{proof}

In view of Lemma \ref{lem:dom}, the maximal expected utility in the frictionless market $S^Z$ associated to \emph{any} supermartingale-CPS $Z\in \mathcal{Z}_{\sup}$ is greater than or equal to its counterpart in the original market with transaction costs:
$$\sup_{f \in \mathcal{A}^{x,ss}_T}\mathbb E[U(f)]\leq\sup_{f \in \mathcal{A}^{x,ss}_T(Z)}\mathbb E[U(f)].$$
The natural question that arises here, and which we address in the sequel, is whether we can find some particularly unfavorable $Z\in\mathcal Z_{\sup}$ such that this inequality becomes an equality.

\begin{definition}	\label{def:shadow}
Fix an initial endowment $x\in \mathbb{R}^d_+ \setminus \{0\}$. The process $S^Z=Z/Z^1$ corresponding to some $Z\in\mathcal Z_{\sup}$ is called a \emph{shadow price process}, if
$$\sup_{f \in \mathcal{A}^{x,ss}_T}\mathbb E[U(f)]=\sup_{f \in \mathcal{A}^{x,ss}_T(Z)}\mathbb E[U(f)].$$
\end{definition}

Some remarks are in order here.

\begin{enumerate}
\item Even if a shadow price exists it need not be unique, cf.\ Remark \ref{rem:counterconstraints}.
\item By Lemma \ref{lem:dom}, any payoff that can be attained in the original market with transaction costs can be dominated in frictionless markets with prices evolving within the bid-ask spread. Hence, strict concavity implies that the optimal payoff $\hat f$ must be the same for a shadow price as in the transaction cost market. In order not to yield a strictly higher utility in the shadow market, the optimal strategy $\hat V$ that attains $\hat f$ with transaction costs must therefore also do so in the shadow market. Put differently, a shadow price must match the bid resp.\ ask prices whenever the optimal strategy $\hat V$ entails purchases resp.\ sales.
\end{enumerate}

\section{Existence of Shadow prices under short selling constraints}

In this section, we prove that a shadow price always exists if short selling is prohibited (cf.\ Corollary \ref{cor:shadow}). To this end, we first derive some sufficient conditions, and then verify that these indeed hold. Throughout, we assume that Assumptions \ref{ass_CPS} and \ref{assfinite} are satisfied.

The following result crucially hinges on the presence of short selling constraints.

\begin{lemma}	\label{supermart} For any supermartingale-CPS $Z\in\mathcal Z_{\sup}$ the following holds:
\begin{enumerate}
\item The process $ZV$ is a supermartingale for any portfolio process $V$ admissible in the sense of Definition \ref{SF}.
\item The process $ZV=Z^1VS^Z$ is a supermartingale for any portfolio process $V$ which is $Z$-admissible in the sense of Definition \ref{SF_shadow}.
\end{enumerate}
\end{lemma}

\begin{proof}
(i) Integration by parts gives
\begin{equation} 	\label{int_parts}
Z_t V_t-Z^1_0 x=\int_0^t V_u dZ_u+\int_0^t Z_u \dot{V}^c_u d\textrm{Var}_u(V^c)+\displaystyle\sum_{u\leq t} Z_{u-}\Delta V_u+\displaystyle\sum_{u< t} Z_{u}\Delta_+ V_u.
\end{equation}
The first integral is a local supermartingale as $V$ is positive, while the other terms are decreasing processes by \cite[Lemma 2.8]{campi.schachermayer.06}. This implies that $ZV$ is a positive local supermartingale and thus a true supermartingale.

(ii) Let $V$ be any $Z$-admissible portfolio process. By \cite[Proposition 2.1]{goll.kallsen.00}, the frictionless self-financing condition \eqref{eq:sf_shadow} is equivalent to the same condition in undiscounted terms, i.e, $d(Z_tV_t)=V_t dZ_t$. Since $ZV$ is positive, it is therefore a positive local supermartingale and hence a supermartingale.
\end{proof}

 The next result presents sufficient conditions for the existence of a shadow price. These are essentially the usual optimality conditions of the duality theory for frictionless markets (cf., e.g., \cite{kramkov.schachermayer.99} and the references therein).

\begin{proposition}\label{suff_cond}
Let $x\in \mathbb{R}^d_+ \setminus \{0\}$. Suppose there are a supermartingale-CPS $Z$ and an a.s.\ strictly positive $\mathcal{F}_T$-measurable random variable $\hat f \in L^0$ satisfying:
\begin{enumerate}
\item $\hat f\in\mathcal A_T^{x,ss}$;
\item $Z^1_T=U'(\hat f)$;
\item $S^{Z,i}_T=Z_T^i/Z_T^1 = 1/\pi_T ^{i1}$ for $i=1,\ldots ,d$;
\item $ \mathbb E[Z^1_T\hat f]=Z_0 x$.
\end{enumerate}
Then, $\hat f$ is the optimal payoff both for the frictionless price process $S^Z$ and in the original market with transaction costs. Consequently, $S^Z$ is a shadow price.
\end{proposition}

\begin{proof} We first prove that $\hat f$ is the optimal solution to the utility maximization problem \eqref{trans_primal} under transaction costs and short selling constraints. By (i), the payoff $\hat f$ is attained by some $\hat V\in\mathcal A^{x,ss}$ after liquidation of $\hat V_{T-}$ at time $T$. Now, take any $X\in\mathcal A^{x,ss}_{T-}$, whose liquidation value to the first asset is $f=X/\pi_{T-}^{\cdot 1}=X/\pi_{T}^{\cdot 1}=XZ_T/Z^1_T$ by (iii) and Assumption \ref{ass:contPi}. Here $1/\pi^{\cdot 1}$ is the vector whose components are given by $1/\pi^{i1}$ for $i=1,\ldots,d$. Then, in view of (iv), we have
$$\mathbb E[Z^1_T\hat f]=Z_0 x
\geq \mathbb E[Z_T X]=\mathbb E[Z^1_T f].$$
Here, the inequality follows from the supermartingale property of $Z$ and $ZV$ for an admissible $V$ leading to $X$ at time $T-$, since Fatou's lemma and the positivity of $Z, V$ imply 
$$\mathbb{E}[Z_T X]=\mathbb{E}[Z_T V_{T-}] \leq \liminf_{t \uparrow T} \mathbb{E}[Z_T V_t] \leq  \liminf_{t \uparrow T} \mathbb{E}[Z_t V_t] \leq Z_0 x.$$

Concavity of $U$ together with Property (ii) in turn gives
\begin{equation*}	\label{polar}
\mathbb E[U(\hat f)-U(f)]\geq \mathbb E[Z^1_T(\hat f-f)]\geq 0,
\end{equation*}
implying that $\hat f$ is the optimal solution to the transaction cost problem \eqref{trans_primal}.

Next, we prove that $\hat f$ is also optimal in the frictionless market with price process $S^{Z}$. To this end, take a strategy $V$ which is $Z$-admissible in the sense of Definition \ref{SF_shadow} with $V_0=x$, so that the portfolio value at time $t$ is given by $W_t:=V_tS_t^{Z}$. The process $Z^1_t W_t=Z_t V_t$ is a supermartingale by Lemma \ref{supermart}(ii), hence we have $\mathbb E[Z^1_T W]\leq Z_0 x$ for any $W\in\mathcal A_T^{x,ss}(Z)$. Now note that the optimization problem in this frictionless market,
$$\sup_{W \in \mathcal{A}^{x,ss}_T(Z)}\mathbb E[U(W)],$$
is dominated by the static problem
\begin{equation*}
\sup \{\mathbb E[U(W)]: W\in L^0 (\mathbb R_+), \mathbb E[Z^1_T W]\leq Z_0 x\},\\
\end{equation*}
which by monotonicity of $U$ can be written as
\begin{equation*}
\sup \{\mathbb E[U(W)]: W\in L^0 (\mathbb R_+), \mathbb E[Z^1_T W]= Z_0 x\}.\\
\end{equation*}
This problem admits a solution which is -- recalling that $I=(U')^{-1}$ -- given by
$$\hat W:=I(Z^1_T)=\hat f=\hat V_{T-}Z_T/Z_1=\hat V_{T-} S^{Z}_T .$$
Indeed, the definition of the conjugate function gives
\begin{equation}    \label{polar_ineq}
\mathbb E[U(W)]\leq \mathbb E[U^*(Z^1_T)]+Z_0 x
\end{equation}
for all random variables $W$ satisfying $\mathbb E[Z^1_T W]= Z_0 x$. The random variable $\hat W=I(Z^1_T)$ attains the supremum (pointwise) in the definition of $U^*(Z^1_T)$ and moreover, by assumption, $\mathbb E[Z^1_T \hat W]= Z_0 x$. Hence, \eqref{polar_ineq} becomes an equality and it follows that $\hat W=\hat f$ is also an optimal payoff in the frictionless market with price process $S^Z$.
The latter therefore is a shadow price as claimed.
\end{proof}

Using the sufficient conditions from Proposition \ref{suff_cond}, we now establish the existence of a shadow price in our multi-currency market model with transaction costs. We proceed similarly as in \cite{loewenstein.00}, adapting the arguments to our more general setting, and using that we have shown existence of an optimal solution $\hat f$ to the utility maximization problem in Proposition \ref{prop:existence} above. First, notice that the value function $J(x)$ is concave, increasing and finite for all $x$ in the set $\mathbb R^d_+ \setminus \{0\}$. By \cite[Theorem 23.4]{rockafellar.70} it is also superdifferentiable for all $x\in \mathbb R^d_+ \setminus \{0\}$.\footnote{In fact, $J$ can be seen as the restriction to $\RR^d_+ \setminus \{0\}$ of some other concave function defined on $K_0$ that allows for negative initial endowment (but forces the agent to make an instantaneous trade at time $0$ in that case).} We recall that the superdifferential $\partial \varphi (x)$ of any concave function $\varphi$ at some point $x$ is defined as the set of all $y \in \mathbb R^d$ such that
\[  \varphi (z) \leq \varphi (x) + y(z-x), \quad \textrm{for all } z \in \mathbb R^d.\]

\begin{proposition}
Fix $x\in\mathbb{R}^d_+ \setminus \{0\}$, the associated optimal solution $\hat f$, and take $h = (h_1 ,\ldots, h_d)$ in the superdifferential $\partial J(x)$. Then, the following properties hold:
\begin{enumerate}
\item $h_1\geq \mathbb E[U'(\hat f)]$;
\item $h_i\geq \mathbb E\left[U'(\hat f)/\pi_T^{1i}\right]$ for $i=2,\ldots,d$;
\item $h\in K^*_0$;
\item $h x=\mathbb E[U'(\hat f)\hat f]$.
\end{enumerate} \label{supergradient}
In particular, the optimal payoff $\hat f$ is a.s.\ strictly positive.
\end{proposition}
\begin{proof}
For $\epsilon>0$ we have $J(x+e_1\epsilon)\geq \mathbb E[U(\hat f+\epsilon)]$ because one can just hold the extra endowment in asset $1$. Hence, the definition of the superdifferential gives
$$h_1\geq \frac{J(x+e_1\epsilon)-J(x)}{\epsilon}\geq \mathbb E\left[\frac{U(\hat f+\epsilon)-U(\hat f)}{\epsilon}\right].$$
Since $U$ is concave, the monotone convergence theorem yields $h_1\geq \mathbb E[U'(\hat f)]$.  In view of the Inada condition $\lim_{x \downarrow 0} U'(x)=\infty$, this also shows that  $\hat f$ is a.s.\ strictly positive.

For $i = 2,\ldots, d$, we have $J(x+e_i\epsilon)\geq \mathbb E[U(\hat f+\epsilon/\pi_T^{1i})]$ because one can hold the extra endowment in asset $i$ and then liquidate it to asset $1$ at time $T$. Hence, as before, we find $h_i\geq \mathbb E[U'(\hat f)/\pi_T^{1i}]$.

Now notice that, for any $i,j=1,\ldots,d$, one can exchange $\pi^{ij}_0$ units of asset $i$ for $1$ unit of asset $j$ at time zero. Hence, $J(x+e_i\epsilon)\geq J (x+e_j\epsilon/\pi_0^{ij} )$, and the definition of the superdifferential yields
$$0 \leq J(x+e_i\epsilon)-J\left(x+ e_j \epsilon/\pi_0^{ij}\right) \leq \epsilon\left(h_i-h_j/\pi^{ij}_0\right).$$
Together with $h_i \geq 0$, $i=1,\ldots,d$, we obtain $h\in K^*_0$.

Finally,
$$h x(\lambda-1)\geq J(\lambda x)-J(x)\geq \mathbb E[U(\lambda\hat f)-U(\hat f)]$$
because $\mathcal A_T^{\lambda x,ss}=\lambda\mathcal A_T^{x,ss}$. Hence if $\lambda>1$, then
$$h x\geq \mathbb E\left[\frac{U(\lambda\hat f)-U(\hat f)}{\lambda-1}\right],$$
and the argument of the expectation increases as $\lambda \downarrow 1$ by concavity of $U$. Analogously, for $\lambda<1$, the inequality is reversed and the argument of the expectation decreases as $\lambda \uparrow 1$. Hence, monotone and dominated convergence yield $h x=\mathbb E[U'(\hat f)\hat f]$ as claimed.
\end{proof}

For any admissible portfolio process $V$, now define the conditional value process
\begin{equation}\label{Jt} J(V,t):=\esssup_{f\in\mathcal A_{t,T}^{V,ss}} \mathbb E[U(f)\mid \mathcal F_t],
\end{equation}
where $\mathcal A_{t,T}^{V,ss}$ denotes the terminal values of admissible portfolio processes which agree with $V$ on $[0,t]$. Let $\hat V$ be the portfolio process in $\mathcal A^{x,ss}$ leading to the optimal solution $\hat f$ to \eqref{trans_primal}. We can apply \cite[Th\'eor\`eme 1.17]{elkaroui.79} to get the following martingale property for the optimal value process $J(\hat V ,t)$ over the whole time interval $[0,T]$:

\begin{lemma}[Dynamic Programming Principle]\label{DPP}
The following equality holds a.s.:
\[
J(\hat V,s)=\mathbb E[J(\hat V ,t)\mid\mathcal F_s],\quad 0\leq s\leq t \leq T.
\]
\end{lemma}

For $i=1,\ldots,d$, now define a process $\tilde Z$ as follows:
$$\tilde Z^i_t:=\lim_{\epsilon\downarrow 0}\frac{J(\hat V+e_i \epsilon,t)-J(\hat V,t)}{\epsilon}, \quad t\in [0,T), \quad \tilde{Z}^i_T :=\frac{U'(\hat V^1_T)}{\pi^{i1}_T}.$$

\begin{proposition}
$\tilde Z$ is a (not necessarily c\`adl\`ag) supermartingale satisfying $\tilde Z_t\in K^*_t$ a.s.\ for all $t\in [0,T]$.
\end{proposition}

\begin{proof}
We adapt the argument of \cite[Lemma 4]{loewenstein.00}. Consider $\epsilon_1,\epsilon_2>0$ with $\epsilon_2<\epsilon_1$. Using the concavity of the utility function $U$ and the definition of the essential supremum yields
\begin{equation*}
\begin{split}
J(\hat V+e_i \epsilon_2,t)&=J\left(\frac{\epsilon_2}{\epsilon_1}(\hat V+e_i\epsilon_1)+\left(1-\frac{\epsilon_2}{\epsilon_1}\right)\hat V,t\right)\\
&\geq \frac{\epsilon_2}{\epsilon_1}J\left(\hat V +e_i\epsilon_1,t\right)+\left(1-\frac{\epsilon_2}{\epsilon_1}\right)J\left(\hat V ,t\right).
\end{split}
\end{equation*}
As a consequence, $\tilde Z^i_t$ is well-defined as the limit of an increasing sequence. For the remainder of the proof, we drop the superscript ``$ss$'' to ease notation. Since the family $\{\mathbb E[U(f)\mid \mathcal F_t]: f\in\mathcal A_{t,T}^{\hat V +\epsilon e_i}\}$ is directed upwards, \cite[Theorem A.3]{karatzas.shreve.98} allows to write the essential supremum as a limit which is monotone increasing in $n$:
$$J(\hat V+e_i\epsilon,t)=\esssup_{f\in\mathcal A_{t,T}^{\hat V+\epsilon e_i}} \mathbb E[U(f)\mid \mathcal F_t]=\lim_{n\to\infty}\mathbb E[U(f^n)\mid \mathcal F_t],$$
where $(f^n)_{n\geq 0}$ is a sequence of elements of $\mathcal A_{t,T}^{\hat V+e_i\epsilon}$.
As $\mathcal A_{t,T}^{\hat V+e_i\epsilon}\subseteq \mathcal A_{s,T}^{\hat V+e_i\epsilon}$ for $0\leq s\leq t<T$,
\begin{equation*}
\begin{split}
J(\hat V+ e_i\epsilon,s)&=\esssup_{f\in\mathcal A_{s,T}^{\hat V+e_i\epsilon}} \mathbb E[U(f)\mid \mathcal F_s]\\
&\geq \esssup_{f\in\mathcal A_{t,T}^{\hat V+e_i\epsilon}} \mathbb E[U(f)\mid \mathcal F_s]\\
&\geq \mathbb E[U(f^n)\mid \mathcal F_s]=\mathbb E[\mathbb E[U(f^n)\mid\mathcal F_t]\mid \mathcal F_s]
\end{split}
\end{equation*}
for all $n\geq 0$. But then monotone convergence gives
\begin{equation*}
\begin{split}
J(\hat V+e_i\epsilon,s)\geq \lim_{n\to \infty}\mathbb E[\mathbb E[U(f^n)\mid\mathcal F_t]\mid \mathcal F_s]&= \mathbb E[\lim_{n\to \infty}\mathbb E[U(f^n)\mid\mathcal F_t]\mid \mathcal F_s]\\
&=\mathbb E[J(\hat V+e_i\epsilon,t)\mid \mathcal F_s].
\end{split}
\end{equation*}
Now, Lemma \ref{DPP} implies
\begin{equation*}
\begin{split}
\frac{J(\hat V+e_i\epsilon,s)-J(\hat V,s)}{\epsilon}\geq \mathbb E\left[\frac{J(\hat V+e_i\epsilon,t)-J(\hat V,t)}{\epsilon}\Big| \mathcal F_s\right]
\end{split}
\end{equation*}
and the supermartingale property of $\tilde Z$ on $[0,T)$ follows by monotone convergence for $\epsilon \downarrow 0$. In order to verify it for the terminal time $T$ as well notice that, for $0 \leq t < T$,
$$J(\hat V+e_i \epsilon,t)-J(\hat V,t) \geq \mathbb{E}\left[U(\hat V^1_T+\epsilon/\pi^{i1}_T)-U(\hat V^1_T)|\mathcal{F}_t\right]$$
because it is admissible to hold the $\epsilon$ extra units of asset $i$ before liquidating them into $\epsilon/\pi^{i1}_{T-}=\epsilon/\pi^{i1}_T$ units of asset $1$, and $J(\hat V,t)=\mathbb{E}[U(\hat V^1_T)|\mathcal{F}_t]$ by Lemma \ref{DPP}. Then, monotone convergence yields
$$\tilde{Z}^i_t \geq \mathbb{E}[U'(\hat V^1_T)/\pi^{i1}_T |\mathcal{F}_t]=\mathbb{E}[\tilde Z^i_T|\mathcal{F}_t], \quad i=1,\ldots,d,$$
such that $\tilde{Z}$ is indeed a supermartingale on $[0,T]$. In particular, it is finite-valued.

It remains to show that $\tilde Z_t\in K^*_t$ for all $t\in [0,T]$. To this end first fix $t\in [0,T)$ and let $(k^n_l)_{l\geq 0}$ be a partition of $[0,\infty)$ with mesh size decreasing to zero as $n$ increases. Note that, for all $\epsilon >0$, on the set $\{k^n_l<\pi_t ^{ij}\leq k^n_{l+1}\}$ we have
$$J(\hat V+ e_i \epsilon,t) - J(\hat V,t)\geq J\left(\hat V+ e_j \epsilon/k^n_{l+1}  ,t\right) - J(\hat V,t)$$
because it is admissible to exchange the $\epsilon$ extra units of asset $i$ for at least $\epsilon/k^n_{l+1}$ units of asset $j$ immediately after time $t$. Again using monotone convergence, this in turn implies
$$\tilde Z^i_t k^n_{l+1}\one_{\{k^n_l<\pi_t ^{ij}\leq k^n_{l+1}\}}\geq \tilde Z^j_t\one_{\{k^n_l<\pi_t ^{ij}\leq k^n_{l+1}\}},$$
and thus
$$\tilde Z^i_t \sum_{l\geq 0} k^n_{l+1}\one_{\{k^n_l<\pi_t ^{ij}\leq k^n_{l+1}\}}\geq \tilde Z^j_t.$$
Then, letting $n\to\infty$ we obtain $\tilde Z^i_t \pi_t ^{ij}\geq \tilde Z^j_t$ for all $i,j=1,\dots,d$. Hence, $\tilde Z_t\in K^*_t$ for $t\in [0,T)$. For the terminal time $T$, this follows directly from the definition and property (iii) in the definition of a bid-ask matrix.
\end{proof}

The process $\tilde Z$ constructed above is a supermartingale but not necessarily c\`adl\`ag. Therefore, we pass to the regularized c\`adl\`ag process $\hat Z$ defined by $\hat Z_T=\tilde Z_T$ and
$$\hat Z^i_t:=\lim_{s\downarrow t, s\in \mathbb Q} \tilde Z^i_s$$
for all $i=1,\ldots ,d$ and $t\in [0,T)$. Note that the limit exists by \cite[Proposition 1.3.14(i)]{karatzas.shreve.88}.

We can now establish our main result, the existence of shadow prices under short selling constraints subject only to the existence of a supermartingale strictly consistent price system (Assumption~\ref{ass_CPS}) and finiteness of the maximal expected utility (Assumption~\ref{assfinite}).

\begin{theorem} \label{existence_shadow}
The process $\hat Z$ belongs to $\mathcal Z_{\sup}$. Moreover, it satisfies the sufficient conditions of Proposition \ref{suff_cond}. Consequently, $S^{\hat Z} = \hat Z/\hat Z^1$ is a shadow price process.
\end{theorem}

\begin{proof}
By \cite[Proposition 1.3.14(iii)]{karatzas.shreve.88}, the process $\hat Z$ is a c\`adl\`ag supermartingale. Moreover, since the bid-ask matrix is right continuous, we have $\hat Z\in\mathcal Z_{\sup}$. By definition, we have $\tilde Z^1_T=U'(\hat f)$ and $\tilde Z_T^i / \tilde Z_T^1= 1/\pi_T ^{i1}$ for $i=1,\ldots ,d$. Since  $\tilde Z$ and $\hat Z$ are equal in $T$, it therefore remains to verify condition (iv) in Proposition \ref{suff_cond}. By Proposition \ref{supergradient},
\begin{equation} \label{shadow_0}
h x=\mathbb E[U'(\hat f)\hat f]=\mathbb E[\hat Z^1_T\hat V^1_T]=\mathbb E[\hat{Z}_T \hat V_T],
\end{equation}
for the portfolio process $\hat V$ attaining $\hat f$. The definition of the superdifferential then gives
$$h_i\geq \frac{J(x+e_i\epsilon)-J(x)}{\epsilon}$$
for any $\epsilon>0$. Hence, $h_i\geq \tilde Z^i_0\geq \tilde Z^i_{0+}=\hat Z^i_{0}$ for all $i=1,\ldots ,d$ by \cite[Proposition 1.3.14(ii)]{karatzas.shreve.88}. Combined with \eqref{shadow_0} and because $x$ has positive components, we obtain
$$\mathbb E[\hat Z_T \hat V_T]=hx\geq \hat Z_0 x$$
Conversely, since $\hat Z\in\mathcal Z_{\sup}$ we can apply the supermartingale property established in Lemma \ref{supermart} which gives $\mathbb E[\hat Z_T\hat V_T]\leq \hat Z_0 x$ and hence $\mathbb E[\hat Z_T \hat V_T]=\hat Z_0 x$. Thus, the sufficient conditions in Proposition \ref{suff_cond} are satisfied and the proof is completed.
\end{proof}

As a result, we can now formulate a precise version of Theorem \ref{mainresult} from the introduction:

\begin{corollary}\label{cor:shadow}
Under short selling constraints and subject to Assumptions \ref{ass_CPS} and \ref{assfinite}, a shadow price in the sense of Definition \ref{def:shadow} exists.
\end{corollary}

\section{Conclusion}\label{conclusion}
We have shown that shadow prices always exist in the presence of short selling constraints, even in general multi-currency markets with random, time-varying, and possibly discontinuous bid-ask spreads. On the other hand, we have presented a counterexample showing that existence generally does not hold beyond finite probability spaces if short selling is permitted. Yet, in simple concrete models the presence of short selling does not preclude the existence of shadow prices, compare \cite{gerhold.al.11}. It is therefore an intriguing question for future research to identify additional assumptions on the market structure that warrant their existence. We also conjecture that shadow prices should always exist for utilities defined on the whole real line, where there is no solvency constraint that can become binding as in our counterexample. Settling this issue, however, will require to resolve the ubiquitous issue of admissibility, potentially along the lines of \cite{biagini.cerny.11}, and is therefore left for future research.

\bibliographystyle{acm}
\bibliography{shadow}

\end{document}